\documentclass[letterpaper,11pt]{article}  
\usepackage{amsfonts}
\usepackage[pdftex]{graphicx}
\usepackage{amsmath,amssymb,amsthm,bm}
\usepackage{natbib}
\usepackage{hyperref}
\usepackage{bm}
\usepackage{fancyhdr}
\usepackage{sectsty}
\usepackage{stmaryrd}
\usepackage{setspace}                      
\usepackage{booktabs}    
\usepackage{pdfsync}        
\usepackage[backgroundcolor=blue!40,linecolor=blue!40]{todonotes}
\usepackage{fancybox}

\newtheoremstyle{myplain}
  {9pt}
  {9pt}
  {\itshape}
  {\parindent}
  {\scshape}
  {:}
  {.5em}
  {}
\newtheoremstyle{mydefinition}
  {9pt}
  {9pt}
  {\itshape}
  {\parindent}
  {\scshape}
  {:}
  {.5em}
  {}
\newtheoremstyle{myremark}
  {9pt}
  {9pt}
  {}
  {\parindent}
  {\scshape}
  {:}
  {.5em}
  {}
\theoremstyle{myplain}
\newtheorem{theorem}{Theorem}
\newtheorem*{theorem*}{Theorem}

\theoremstyle{mydefinition}

\theoremstyle{myremark}

\renewcommand{\cite}{\citet}

\newcommand{\cR}{\mathcal{R}} 

\newcommand{\cT}{\mathcal{T}}

\hypersetup{colorlinks=true, linkcolor=blue, citecolor=blue}                          
\numberwithin{equation}{section}

\begin{document}

\title{Revealed Stochastic Preference:\\A One-Paragraph Proof and Generalization}
\date{January 16th, 2019}
\author{J\"{o}rg Stoye\thanks{Department of Economics, Cornell University, stoye@cornell.edu. This note was motivated by joint work with Yuichi Kitamura. Thanks to an anonymous referee for helpful comments. Financial support through NSF Grants SES-1260980 and SES-1824375 is gratefully acknowledged.}}

\maketitle

\begin{abstract}
\small
\cite{mcfadden-richter} and later \cite{McFadden05} show that the \textit{Axiom of Revealed Stochastic Preference} characterizes rationalizability of choice probabilities through random utility models on finite universal choice spaces. This note proves the result in one short, elementary paragraph and extends it to set valued choice. The latter requires a different axiom than is reported in \cite{McFadden05}.
\medskip
\small


\end{abstract}

\section{Introduction}

\cite{mcfadden-richter} and \cite{McFadden05} show that, for a finite choice universe and singleton-valued choice, the \textit{Axiom of Revealed Stochastic Preference} (ARSP) is necessary and sufficient for rationalizability of stochastic choice through a random utility model. This note proves the same result in one short and ``chatty" paragraph. Indeed, setting up an efficient notation  makes the result so immediate that one might hesitate to report it. However, the precise argument appears to have been overlooked.\footnote{The proofs in \cite{mcfadden-richter} and \cite{border} take several pages. The proof in \cite{McFadden05} is one dense paragraph but goes through several intermediate characterizations and invokes duality, though some of that is to also establish other characterizations. \cite[][p.123]{GulPesendorfer06} call the ARSP and related axioms ``complicated restrictions on arbitrary finite collections of decision problems" that are ``difficult to interpret." They establish the result by relating ARSP to their own axioms. Within the considerable literature inspired by \cite{Falmagne}, the basic geometry of random utility models is folk knowledge, but the focus is on other questions.} Furthermore, it immediately applies to set-valued choice. The resulting generalization is novel; in particular, if ``decisiveness" of choice were dropped in \cite{McFadden05}, the axiom reported there would be too weak.

\section{The Classic Setting and Result}

Consider a set of choice problems defined as subsets of a finite choice universe.\footnote{This setting encompasses a classic demand setting with finitely many budgets because all revealed preference information is then contained in choice probabilities corresponding to a certain finite partition of the choice universe. See \cite{McFadden05} and \cite{KS}.} Suppose a researcher observes the corresponding choice probabilities. The question is whether these probabilities can be rationalized through maximization of a random utility function. The random utility could reflect randomness in an individual's utility assessment but also unobserved heterogeneity if data are from a repeated cross-section. Also, it will become clear that the underlying notion of nonstochastic utility is very flexible.

It will be useful to formalize everything in terms of vectors. Thus, identify the finite choice universe with components of a vector $X=(x_1,\dots,x_K)$. Any choice problem $C$ is a subvector of $X$. For any observable $C$, the corresponding choice probabilities will be collected in a vector $\pi(C) \in \Delta^{|C|-1}$, the $(|C|-1)$-dimensional unit simplex. Let $\mathcal{C}=(C_1,\dots,C_J)$ collect all observable choice problems in arbitrary (but henceforth fixed) order. We can then represent the totality of observed choice probabilities by the vector $\Pi \equiv (\pi(C_1),\dots,\pi(C_J))$. Note that $\Pi \in [0,1]^I$, where $I \equiv \sum_{j=1}^J{\vert C_j\vert}$, and that the components of $\Pi$ sum to $J$. 

A set of particular interest will be the collection of rationalizable \textit{nonstochastic} choice behaviors, i.e. the set of rationalizable choice functions on $\mathcal{C}$. For example, rationalizability could be defined through independence of irrelevant alternatives for individual choice \citep{Arrow59} or (in a demand context) the Generalized Axiom of Revealed Preference \citep{Afriat67}; in either case, the nonstochastic rationality concept then is utility maximization. But ``rationalizability" is used in a very broad sense -- one could easily define more narrow, more lenient, or also nonnested criteria.\footnote{See \cite{DKQS} for a recent random utility model of the latter kind.} Any choice function can be expressed as a vector $R \in \{0,1\}^I$ that is interpreted just like $\Pi$ but whose entries are binary, with exactly one entry of $1$ for each choice problem. Let the set $\mathcal{R}$, which is necessarily finite, collect the vector representations of all rationalizable choice functions.

We call $\Pi$ \textit{stochastically rationalizable} if it can be generated by the ``random utility" extension of the notion of rationalizability embodied in $\mathcal{R}$, that is, if the choice probabilities $\Pi$ can be generated by drawing a nonstochastic ``choice type" $R$ from some (arbitrary but fixed) distribution on $\mathcal{R}$ and then executing this type's nonstochastic choices. In short, the set of stochastically rationalizable choice probabilities $\Pi$ is the convex hull of $\mathcal{R}$, a finite polytope that is henceforth denoted $\text{co}(\mathcal{R})$.

\cite{mcfadden-richter} characterize this notion of rationalizability through the ARSP, an axiom about sequences of subsets of choice problems and corresponding choice probabilities. The present notation allows for a succinct statement. To this purpose, let a \textit{trial} $T \in \{0,1\}^I$ be a vector whose ``$1$" entries correspond to (some or all) elements of the same choice problem. Intuitively, any trial characterizes a specific subset of a specific choice problem. (If the same subset of $X$ appears as subset of distinct choice problems, this gives rise to distinct trials.) Let the finite set $\cT$ collect all possible trials. Then we have:
\medskip

\textbf{Axiom of Revealed Stochastic Preference (ARSP)}

For any finite sequence $(T_1,\dots,T_M) \in \cT^M$, $M \leq \infty$, one has
\begin{equation*}
\sum_{m=1}^M{T_m \Pi} \leq \max_{R \in \mathcal{R}}{\sum_{m=1}^M{T_m R}}. %\label{eq:ARSP}
\end{equation*}
\medskip

In words, the choice probabilities corresponding to any finite sequence of trials cannot add up to more than the maximal analogous sum induced by a nonstochastic choice function in $\cR$. We then have:

\begin{theorem} \label{T1} 
$\Pi$ is stochastically rationalizable if, and only if, ARSP holds.
\end{theorem}

\begin{proof}
The sequence $(T_1,\dots,T_M)$ matters only through $\sum_{m=1}^M T_m$. This sum is an $I$-vector with integer components and (because the canonical basis of $\mathbb{R}^I$ is in $\cT$) any $I$-vector with integer components can be expressed as such a sum. Hence, ARSP can be rewritten as
\begin{eqnarray*}
&& T\Pi \leq \max_{R \in \mathcal{R}}{TR}~~\text{for all vectors}~T \in \mathbb{R}_+^I~\text{with integer components} \\
&\overset{(1)}{\Leftrightarrow} & T\Pi \leq \max_{R \in \mathcal{R}}{TR}~~\text{for all vectors}~T \in \mathbb{R}_+^I~\text{with rational components} \\
&\overset{(2)}{\Leftrightarrow} & T\Pi \leq \max_{R \in \mathcal{R}}{TR}~~\text{for all vectors}~T \in \mathbb{R}_+^I \\
&\overset{(3)}{\Leftrightarrow} & T\Pi \leq \max_{R \in \mathcal{R}}{TR}~~\text{for all vectors}~T \in \mathbb{R}^I \\
&\overset{(4)}{\Leftrightarrow} & \Pi \in \text{co}(\cR).
\end{eqnarray*}
Here, $(1)$ holds because inequalities can be multiplied by positive scalars; $(2)$ holds because weak inequalities are preserved under limit taking and all vectors in $\mathbb{R}_+^I$ can be approximated by rational ones. To see $(3)$, let $\bm{1} \equiv (1,\dots,1) \in \mathbb{R}^I$, then $\bm{1}\Pi=J$ but also $\bm{1}R=J$ for any $R \in \mathcal{R}$. Thus, $$T\Pi \leq \max_{R \in \mathcal{R}}{TR} \Leftrightarrow (T+\bm{1} \cdot \Vert T \Vert_\infty)\Pi \leq \max_{R \in \mathcal{R}}{(T+\bm{1} \cdot \Vert T \Vert_\infty)R},$$ and $T+\bm{1} \cdot \Vert T \Vert_\infty \in \mathbb{R}^I_+$ even if $T$ is not. (Recall that $\Vert T \Vert_\infty$ is the largest absolute value taken by any component of $T$.) Finally, $(4)$ is an elementary property of convex hulls.
\end{proof}

The theorem is not new but the proof reflects two contributions. First, its almost embarrassing simplicity is partly due to efficient notation, i.e. to recognizing that with a finite universal choice space, all important quantities can be expressed as vectors and choice probabilities then as inner products. That ARSP prevents $\Pi$ from being separable from $\text{co}(\cR)$ in any positive direction is then near immediate and has also been anticipated in the literature. The second contribution is to notice that adding-up constraints on $\Pi$ and $R$ relate separation in \textit{any} direction to separation in a positive direction.

The result's exposition also renders immediate some (previously known) facts about ARSP and $\text{co}(\cR)$. For example, while there are infinitely many conceivable sequences $(T_1,...,T_M)$, a finite subset of ``essential" sequences suffices to characterize stochastic rationalizability. In particular, the essential sequences correspond to gradients of full-dimensional faces of $\text{co}(\cR)$. These gradients are rational because $\text{co}(\cR)$ is spanned by binary vectors; therefore, the essential sequences are indeed finite. The gradients are known in some special cases: For binary choice and $K \leq 5$, they correspond to the Block-Marschak triangle inequalities \citep{Dridi1980}, and they can be explicitly computed for the consumer choice problem with finitely many goods and $3$ budgets \citep[][display 4.7]{KS} or with finitely many budgets and $2$ goods \citep{HS15}. However, computing them is equivalent to deriving the halfspace representation of the polytope $\text{co}(\mathcal{R})$ from its vertex representation, and this transformation is in general computationally prohibitive.\footnote{See \cite{Ziegler} for more on these representations. The difficulty of this computation is illustrated by the intricate literature on finding facet defining inequalities for the linear order polytope \citep{Fishburn92}.}

\section{Extension to Set Valued Choice}

With a finite choice universe, the above analysis is easily extended to choice correspondences, i.e. set-valued choice functions. Specifically, define a new choice universe $\tilde{X}\equiv 2^X$, the subsets of $X$. Any choice problem $C$ in the original problem can be identified with a new problem $\tilde{C} \equiv 2^C$. Then any data set concerning set valued choice on $(C_1,...,C_J)$ can be interpreted as data set with unique choice on $(\tilde{C}_1,...,\tilde{C}_J)$. A choice probability $\tilde{\Pi}$, a collection $\tilde{\cR}$ of rationalizable nonstochastic behaviors, and a set $\tilde{\cT}$ of possible trials can be defined as before. Asserting the ARSP for these new quantities yields:
\medskip

\textbf{Generalized Axiom of Revealed Stochastic Preference (GARSP)}

For any finite sequence $(\tilde{T}_1,\dots,\tilde{T}_M) \in \tilde{\cT}^M$, $M \leq \infty$, one has
\begin{equation*}
\sum_{m=1}^M{\tilde{T}_m \tilde{\Pi}} \leq \max_{\tilde{R} \in \tilde{\cR}}{\sum_{m=1}^M{\tilde{T}_m \tilde{R}}}.
\end{equation*}
\medskip

Then the previous proof establishes:

\begin{theorem} \label{T2}
$\tilde{\Pi}$ is stochastically rationalizable if, and only if, GARSP holds.
\end{theorem}

Trivial as it may seem, this extension appears novel and contains two insights.
\begin{itemize}
\item Contrary to most of the literature, the empty set is included in the choice universe and all choice problems. Thus, Theorem \ref{T2} applies to the stochastic extension of incomplete preference models and the like. If the underlying nonstochastic choice model excludes the empty set, this will be reflected in $\tilde{\cR}$ and, therefore, in all rationalizable $\tilde{\Pi}$.\footnote{The same applies if $\tilde{\cR}$ forces choice of singletons, in which case Theorem \ref{T1} is recovered.}
\item \cite{mcfadden-richter} restrict attention to singleton-valued choice ``for simplicity" (footnote 2); \cite{McFadden05} presents an extremely general setup and passes to singleton-valued choice only after defining ARSP. While neither formally characterize set-valued stochastic choice, the reader might infer that, for example, Theorem 3.3 in \cite{McFadden05} holds with ARSP as reported there even if the choice function is not singleton valued. This is not the case. Both papers define choice probability as the probability that a set \textit{or any of its subsets} are chosen from the corresponding budget. This does not lose information and makes no difference at all for singleton-valued choice. But with set-valued choice, asserting ARSP for this probability restricts trials to query a set \textit{and all its subsets} within a specific choice problem. This is a small subset of the trial sequences that are actually needed, so the restriction to single-valued choice is crucial and the axiom is otherwise too weak.\footnote{Consider any setting where at least one choice problem is not a singleton, any $\tilde{\cR}$ that restricts choice to be singleton-valued, and the choice probability $\tilde{\Pi}$ that always selects the entire choice problem. For any trial that queries an entire choice problem and all its subsets, $\tilde{\Pi}$ as well as any $\tilde{R} \in \tilde{\cR}$ return value $1$. For any other trial, $\tilde{\Pi}$ returns $0$. Thus ARSP as defined in \cite[][section 2.3.3]{McFadden05} is fulfilled, but $\tilde{\Pi} \notin \text{co}(\cR)$. For that paper's Theorem 3.3 to apply to set-valued choice, choice probabilities should be defined as explained above and the set inclusions in expressions (1) and (2) should be equalities.} 
\end{itemize}
 
\section{Conclusion}
This note concisely proves that rationalizability of stochastic choice is characterized by (G)ARSP. Most of the work is done by a vector notation which obviates that the axiom really restricts separating hyperplanes. Beyond its simplicity, the approach (i) further clarifies some known facts about stochastic rationalizability and (ii) immediately generalizes to set-valued choice, where it allows for indecisive (empty valued) choice and clarifies the appropriate axiom.
\bibliography{KMS}

\end{document}